\documentclass[10pt,twocolumn]{article}
\usepackage[utf8]{inputenc}
\usepackage[T1]{fontenc}
\usepackage{amsmath, amssymb, amsthm}
\usepackage{graphicx}
\usepackage{xcolor}
\usepackage{hyperref}
\usepackage{comment}
\usepackage{tabularx}

\usepackage[margin=2cm]{geometry}

\newcommand{\DME}{\textsc{DME}}

\newcommand{\NOP}{\text{NOP}}

\newtheorem{theorem}{Theorem}
\newtheorem{corollary}{Corollary}[theorem]

\title{Divergent Multi-Version Execution (DME):\\
Canonical Instruction-Trace Fault Detection via Structural Address-Space Decorrelation}

\author{Petro Baran Yr\\
Independent Researcher, Uzhgorod\\
\small 2026}

\date{}

\begin{document}

\maketitle

\begin{abstract}

Traditional redundancy (lockstep, TMR) executes identical binaries with identical memory layouts. A single correlated fault — for example, an arbitrary program counter value or a perturbation \(\Delta PC\) in all replicas — redirects all replicas along the same incorrect path. The same applies to corruption of data pointers. Both types of faults, regardless of their origin (deliberate tampering, software bug, compilation bug, or physical disturbance), cause silent data corruption and erroneous program execution.
This work presents Divergent Multi-Version Execution (DME) — a \textbf{runtime semantic consistency verifier} for diversified executions.  Each replica is compiled independently, producing different code and data memory layouts while preserving identical semantics. Faults are detected by comparing canonical instruction traces, which include opcodes, register identifiers, loaded/stored values, and results — while discarding layout-dependent addresses.

Under fault-free execution, all replicas produce identical canonical traces. Any fault related to erroneous code or data addresses causes trace divergence and fault detection.

\paragraph{Key features:}
\begin{itemize}
    \item Semantic execution trace monitoring.
    \item Structural correctness monitoring during execution.
\end{itemize}

Depending on the chosen detection method (semantic comparison only, or combined with structural correctness monitoring), DME provides either deterministic or probabilistic fault detection guarantees.

\paragraph{Deterministic detection:}
\begin{itemize}
    \item Instruction pointer corruption is guaranteed detectable within one instruction within the fine-grained NOP-decorrelated region.
    \item \(PC_1 = PC_2 = \dots = PC_N = \text{const}\) error — all replicas' program counters become set to the same erroneous value: guaranteed detectable via structural address equality violation.
    \item \(*p_1 = *p_2 = \dots = *p_N = \text{const}\) error — all replicas' data pointers become set to the same erroneous value: guaranteed detectable via structural address equality violation.
    \item Data divergence in one or \(N-1\) replicas: deterministic detection.
\end{itemize}

\paragraph{Probabilistic detection:}
These are faults where program counters or data pointers diverge to different erroneous values across replicas (e.g., \(PC_1 \neq PC_2 \neq \dots \neq PC_N\) or corruption occurs in only one replica or \(N-1\) replicas). Such faults are not covered by the deterministic regime and are detected with probability bounded by exponential decay. For typical 32-bit ISAs and a 4\,kB memory region, the per-step probability of remaining undetected is bounded by \(\varepsilon \leq 2^{-52}\).
This upper bound is determined statically from the program structure, including the control-flow graph, instruction distribution, and memory layout. Consequently, probabilistic fault-detection guarantees can be evaluated at compile time for a given fault model, enabling verification that the resulting protection satisfies the required safety or security targets.

\end{abstract}

\noindent
\fbox{%
\parbox{\dimexpr\linewidth-2\fboxsep-2\fboxrule\relax}{%
\textbf{Fault Amplification through Structural Address-Space Decorrelation}
\vspace{2pt}
\hrule
\vspace{2pt}
\textbf{Fundamental principle.} Under fault-free execution, all replicas produce
identical canonical traces; under any fault, structural decorrelation guarantees
maximally divergent outcomes.
}}%

\section{Motivation: Failure of Conventional Redundancy under Correlated Faults}

Conventional redundancy schemes such as lockstep execution and TMR
primarily assume that faults affecting different replicas are independent.

However, experimental studies on electromagnetic fault injection have shown
that strongly correlated instruction-level perturbations are physically
realizable in practice [  ~\ref{Dutertre2021} ]. In particular, Dutertre et al.
demonstrated that a single electromagnetic pulse can skip multiple
consecutive instruction fetches on an 8-bit microcontroller.

In conventional lockstep or TMR systems executing identical binaries with
identical memory layouts, such correlated perturbations may redirect all
replicas along the same incorrect execution path, causing silent data
corruption.

Divergent Multi-Version Execution (DME) is designed to detect not only such
correlated control-flow faults, but more generally perturbations capable of
altering the intended execution semantics of a program, including pointer
corruption, memory-address faults, and misdirected control transfers.

Unlike conventional MVEE systems, DME interprets cross-replica address equality as a violation of structural independence rather than as a permissible execution state. This invariant enables detection of a broad class of faults that manifest identically across all replicas — including data pointer corruptions (e.g., null or value-as-pointer errors), application software bugs (e.g., uninitialized pointers or return address overwrites), and compiler/linking defects that produce correlated layout or addressing anomalies — all of which are detected immediately upon address collapse, without relying on semantic trace divergence.

\section{\DME \  Model}

\subsection{Architectural Summary: Core Principles}

The DME architecture rests on four fundamental principles enforced at compile time and optionally customized by the application developer.

\begin{enumerate}
\item \textbf{Independent Compilation.} All \(N\) replicas are compiled independently from the same source program. This produces distinct code and data layouts while preserving identical opcode-level semantics and isomorphic control-flow graphs. No runtime coordination is required for layout generation.

\item \textbf{Coarse-Grained Decorrelation: Function and Block Interleaving.} Functions and basic blocks are placed at different addresses across replicas using complementary branch displacement patterns. Specifically, inter-block and inter-function jump offsets are assigned opposite signs in different replicas (e.g., forward branches in replica 0 become backward branches in replica 1). This ensures that identical control-flow perturbations map to semantically divergent execution paths.

\item \textbf{Fine-Grained Decorrelation: Deterministic NOP Insertion.} Periodic NOP instructions are inserted at replica-specific offsets:
\[
\mathrm{offset}_r = \frac{l}{N} \cdot r, \quad r = 0,\dots,N-1
\]
where \(l\) is the insertion period (stride) specified by the application developer. The density \(1/l\) controls the trade-off between detection granularity and overhead. The developer may enable NOP insertion either:
\begin{itemize}
\item \textbf{Globally} — across the entire code section, or
\item \textbf{Selectively} — only within critical code regions (e.g., control-flow intensive loops, security-sensitive paths).
\end{itemize}
This flexibility allows tuning of fault detection resolution versus performance cost.

\item \textbf{Data Address Decorrelation.} Global variable addresses are shuffled independently for each replica. Stack frames and heap allocations are placed at different base offsets across replicas. Combined with canonical trace comparison (which excludes absolute addresses), this ensures that identical memory-addressing faults produce different loaded/stored values across replicas, exposing data-path faults.
\end{enumerate}

Together, these principles transform correlated low-level perturbations into observable semantic divergence while maintaining identical program semantics across replicas.

\begin{figure}[htbp]
    \centering
    \includegraphics[width=\columnwidth]{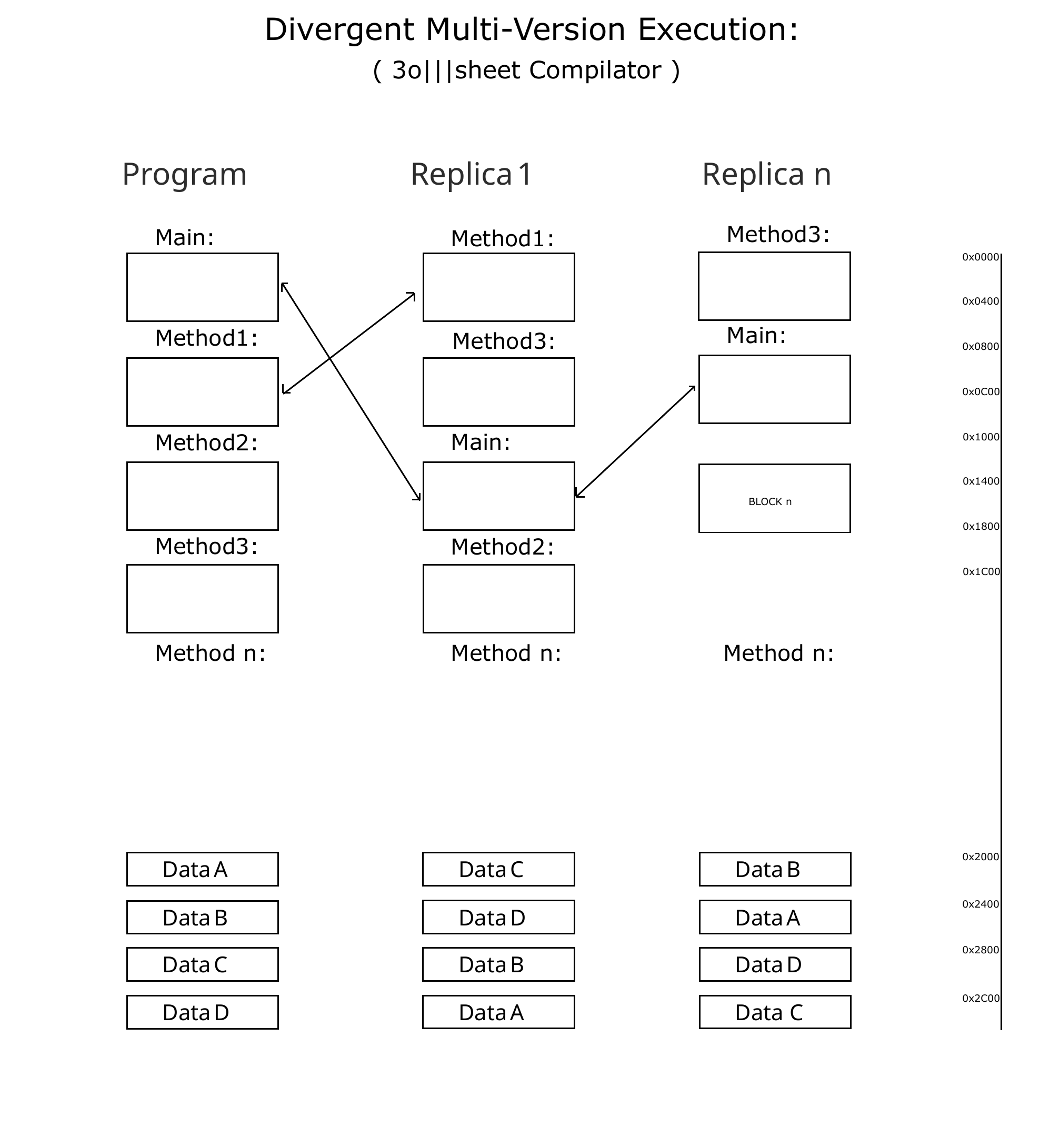}
    \caption{Architecture.}
    \label{fig:collapse}
\end{figure}

\subsection{Replicas}
\(N \geq 2\) replicas, same source program.
\begin{itemize}
\item Each compiled independently \(\rightarrow\) distinct physical layouts
\item Private code, stack, heap per replica
\item Identical opcode-level semantics and control-flow graph (isomorphic)
\end{itemize}


\begin{figure*}
    \centering
    \includegraphics[width=0.7\textwidth]{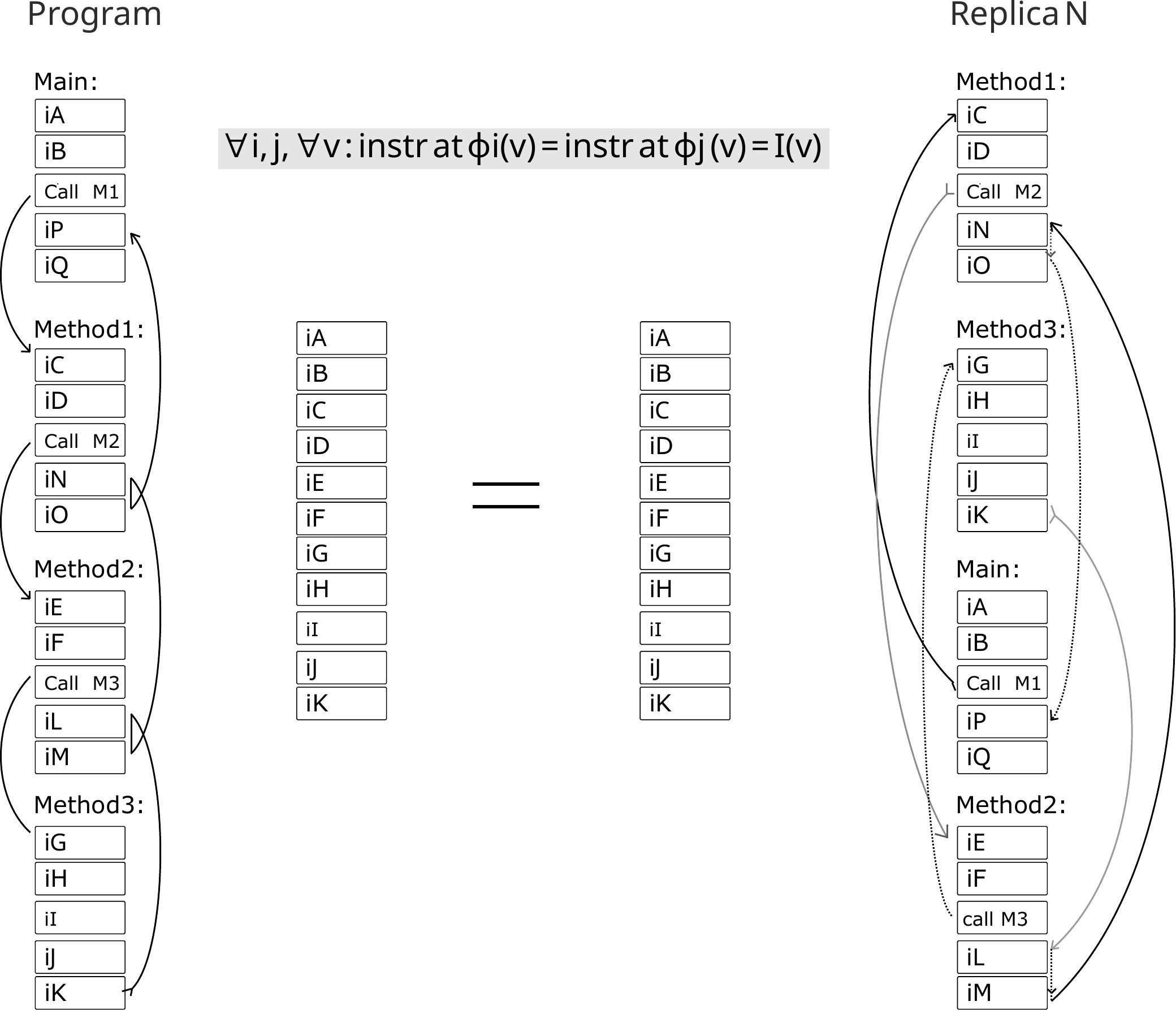}
    \caption{ Normal (fault-free) execution. Each replica's methods, objects, and variables are relocated in memory, with recompilation applying new addresses for instruction use during runtime.}
    \label{fig:scheme}
\end{figure*}



\begin{figure}[htbp]
    \centering
    \includegraphics[width=\columnwidth]{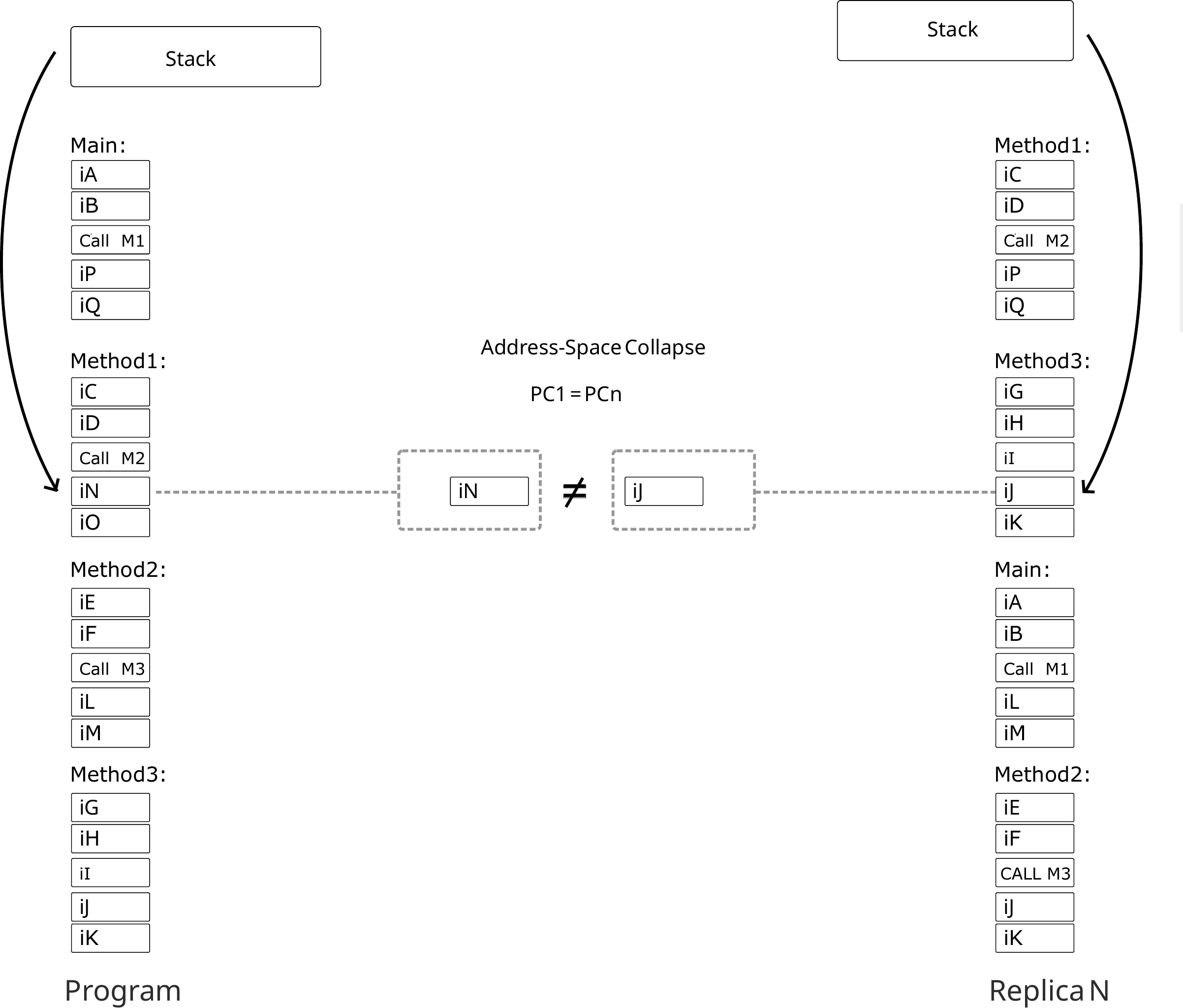}
    \caption{During normal execution, each replica's method has its own return address.
             Stack overflow can direct all replicas to the same address, causing trace divergence:
             $PC_1 = PC_2 = \dots = PC_N$.
             Simultaneous corruption of the return address in all replicas (via buffer overflow) $\rightarrow$ all PCs equal $\rightarrow$ immediate error.
             In N-Variant Systems, this could go undetected if the address falls into the allowed region.}
    \label{fig:collapse}
\end{figure}


\begin{figure}[htbp]
    \centering
    \includegraphics[width=\columnwidth]{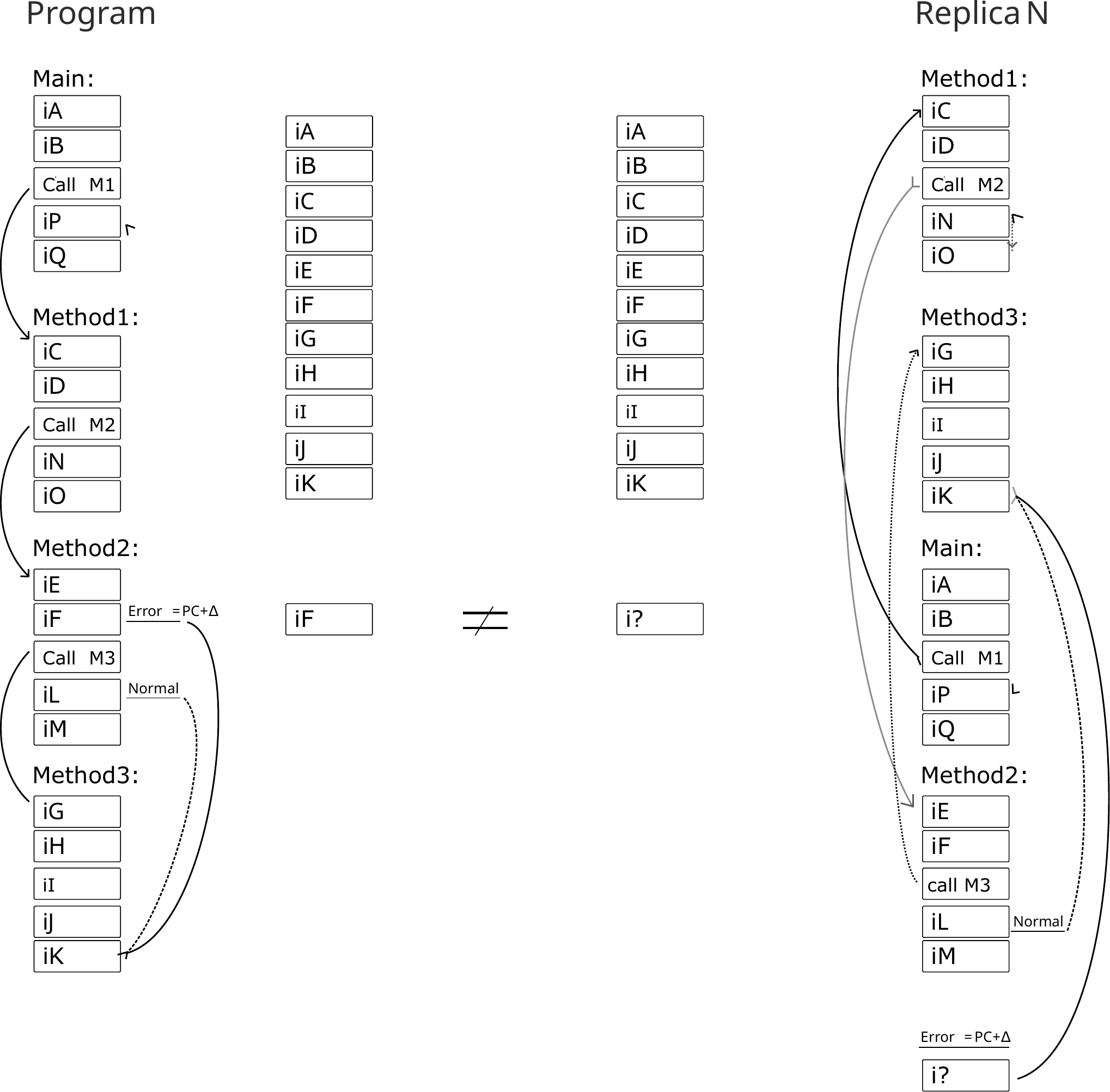}
    \caption{Address corruption in replicas may be partial (e.g., single-bit or multi-bit flips), resulting not in a uniform replacement address but in an offset from the normal value while remaining distinct across replicas. To detect such errors, DME employs opposite branch signs and NOP insertion to shift entry points.}
    \label{fig:faulty}
\end{figure}

\section{Address-Space Decorrelation}

Since DME targets bare-metal embedded systems without virtual memory (no MMU), all addressing is performed using absolute physical addresses. Structural address-space decorrelation is achieved by compiling and linking each replica independently with completely different absolute memory layouts.

For each replica \( r \), the compiler and linker assign unique absolute base addresses to code sections, data sections, stack, and heap regions. As a result, the same source-level function or variable is placed at significantly different absolute memory addresses across replicas.

DME employs three complementary mechanisms to achieve strong address-space decorrelation:

\begin{enumerate}
\item \textbf{Block-level diversification} --- functions and basic blocks are placed at different absolute addresses in each replica. Complementary branch displacement patterns are used (e.g., forward branches in one replica become backward branches in another).

\item \textbf{Instruction-level asymmetry} --- periodic \NOP\ insertion with replica-specific offsets:
\[
\mathrm{offset}_r = \frac{l}{N} \cdot r, \quad r = 0, \ldots, N-1
\]
where \( l \) is the insertion period (stride).

\item \textbf{Data layout diversification} --- global variables, static data, stack frames, and heap allocations are placed at different absolute addresses in each replica. Stack pointers are initialized to replica-specific base addresses.
\end{enumerate}

\subsection{Forced Fragmentation of Oversized Functions / Blocks}

When a function or basic block exceeds a critical size \( L_{\text{crit}} \) (e.g., 256 bytes), deterministic detection granularity degrades and coarse-grained decorrelation weakens due to locally similar branch displacements within the block.

To mitigate this, DME applies \textbf{forced fragmentation} during independent compilation of each replica. For any block exceeding \( L_{\text{crit}} \), the compiler selects fragmentation points. At each point:

\begin{itemize}
    \item The original instruction sequence is split into fragments.
    \item Each fragment is placed at a replica-unique absolute memory location.
    \item The original site receives an unconditional \texttt{JMP} to the fragment.
    \item The fragment ends with a \texttt{JMP} back to the continuation point.
\end{itemize}

\paragraph{Example (\(N=2\)):}
Consider a logical block consisting of instructions \( A, B, C, D, E, F \).

\noindent\textit{Replica 0} (fragments after \( B \) and \( D \)):
\[
\text{Main: } A, B, \texttt{JMP f1} \quad
\text{f1: } C, D, \texttt{JMP f2} \quad
\text{f2: } E, F, \texttt{JMP ret}
\]

\noindent\textit{Replica 1} (fragments after \( C \) and \( E \)):
\[
\text{Main: } A, B, C, \texttt{JMP f1} \quad
\text{f1: } D, E, \texttt{JMP f2} \quad
\text{f2: } F, \texttt{JMP ret}
\]

A correlated fault \( \Delta PC = +2 \) causes the replicas to execute different instructions (\( B \to D \) in Replica 0 vs. \( B \to E \) in Replica 1), leading to immediate divergence of the canonical instruction trace.

\paragraph{Overhead:} 
Each fragmentation point increases code size by 2--4 bytes (JMP pair). To limit overhead, fragmentation can be applied selectively only to control-flow intensive or security-critical regions.


\begin{figure}
    \centering
\includegraphics[width=0.5\columnwidth]{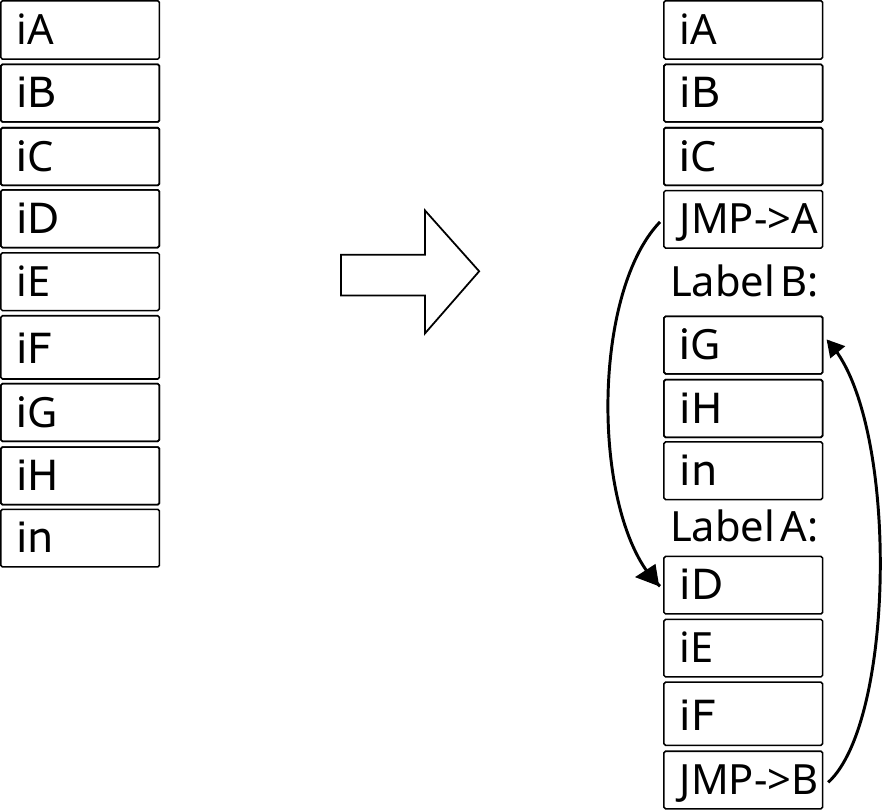}
    \caption{Forced fragmentation of oversized blocks. This technique can optionally complement or replace NOP padding for fine-grained decorrelation.}
    \label{fig:fragmentation}
\end{figure}


\begin{paragraph}{Relationship to DME mechanisms:}
\begin{center}
\small
\begin{tabularx}{\columnwidth}{|X|l|X|}
\hline
Mechanism & Scope & Detection type \\ \hline
NOP insertion & fine-grained & deterministic ($\ge l/N$) \\ \hline
Block interleaving & coarse-grained & probabilistic \\ \hline
\textbf{Fragmentation} & oversized blocks & deterministic / probabilistic \\ \hline
\end{tabularx}
\end{center}
\end{paragraph}

Forced fragmentation ensures that no single contiguous block remains large enough to hide a correlated fault. The deterministic detection bound (Theorem 1) applies independently to each fragment.


\section{Canonical Instruction Trace (Full Definition)}

The canonical instruction representation \(C(I,s)\) for an instruction \(I\) 
and resulting architectural state \(s\) includes:

\begin{itemize}
    \item opcode and condition codes,
    \item source and destination register identifiers,
    \item immediate operands,
    \item loaded memory values (for load instructions),
    \item computed results (for ALU operations),
    \item stored values (for store instructions).
\end{itemize}

Absolute or relative addresses are explicitly excluded, as are NOP instructions.

\textbf{Example.} Two replicas with different NOP layouts:

\begin{verbatim}
Replica 1: A, B, C, NOP, D, E, F, G
Replica 2: NOP, A, B, C, D, E, NOP, F
\end{verbatim}

Under fault-free execution, both produce the same canonical trace:
\( \langle A, B, C, D, E, F, G \rangle \).

If a fault causes \(\Delta PC = +3\) during instruction \(B\):
\begin{itemize}
    \item Replica 1: \(B \to D\) (skips NOP)
    \item Replica 2: \(B \to E\)
\end{itemize}
Canonical traces diverge → fault detected immediately.

\paragraph{Key observation.}
The inclusion of \emph{computed results} in the canonical trace creates a
temporal coupling between consecutive instructions. For an ALU operation,
the result becomes an input to subsequent instructions. Consequently,
if two replicas produce identical canonical traces at time $t$, their
register states must be identical at time $t$ (at least for the
destination registers). By induction, identical traces over a sequence
imply identical full register states throughout that sequence.

Thus, undetected execution requires not only semantically equivalent
instructions but also \emph{identical prior computation history}.
This dramatically reduces the probability of prolonged undetected faults
compared to schemes that compare only opcodes or control-flow signatures.

\subsection{Fault Model}
Control-flow fault:
\[
PC_f^{(r)} = PC_i^{(r)} + \Delta^{(r)}
\]
We consider:
\begin{itemize}
\item Independent faults (single replica affected)
\item Partially correlated faults \((\Delta^{(i)} \neq \Delta^{(j)})\)
\item Fully correlated faults \((\Delta^{(1)} = \Delta^{(2)} = \dots = \Delta)\) — the dangerous case for conventional redundancy.
\end{itemize}


\subsection{Unified Perturbation Model}

Conventional fault-tolerance models distinguish between
hardware faults, software defects, compiler bugs, and
environmental disturbances as separate categories.

DME does not fundamentally depend on the origin of a
perturbation. Instead, any mechanism capable of altering
execution semantics is modeled uniformly as a semantic
perturbation operator:

\[
\Pi : S_t \rightarrow S'_t
\]

where:

\begin{itemize}
    \item \(S_t\) is the expected architectural state at time \(t\),
    \item \(S'_t\) is the perturbed architectural state.
\end{itemize}

The origin of \(\Pi\) is irrelevant to the DME model.

Examples include:

\begin{itemize}
    \item transient hardware faults,
    \item permanent hardware defects,
    \item compiler miscompilation,
    \item linker corruption,
    \item memory corruption,
    \item undefined behavior manifestations,
    \item malicious fault injection,
    \item control-flow hijacking,
    \item voltage or clock glitches,
    \item electromagnetic interference,
    \item radiation-induced bit flips.
\end{itemize}

DME verifies only whether diversified executions remain
semantically equivalent.

Let \(T_r\) denote the canonical execution trace of
replica \(r\). Runtime correctness is defined as:

\[
T_1 = T_2 = \cdots = T_N
\]

Violation of semantic equivalence:

\[
\exists i \neq j : T_i \neq T_j
\]

indicates incorrect execution independently of the
perturbation origin.

Thus, DME should not be viewed solely as a fault
detection mechanism, but more generally as a runtime
semantic consistency verifier for diversified executions.


\subsection{Assumptions and Threat Model}

DME assumes that replicas preserve semantic equivalence
under fault-free execution while maintaining structural
independence of address mappings.

The model assumes:

\begin{itemize}
    \item independent replica compilation,
    \item non-identical code and data layouts,
    \item deterministic canonicalization,
    \item correct synchronization of execution slices,
    \item trusted trace comparison logic.
\end{itemize}

DME targets perturbations capable of modifying
execution semantics, including:

\begin{itemize}
    \item transient hardware faults,
    \item memory corruption,
    \item compiler or linker defects,
    \item control-flow corruption,
    \item pointer corruption,
    \item fault injection attacks.
\end{itemize}

The model does not guarantee detection of perturbations
that preserve semantic equivalence across all replicas.

Examples include:

\begin{itemize}
    \item identical semantic-preserving compiler bugs,
    \item faults affecting shared architectural state
          identically,
    \item perturbations confined entirely to excluded
          shared regions.
\end{itemize}


\section{Detectability Guarantees}

DME provides two complementary regimes of semantic
divergence observability:

\begin{enumerate}
    \item \textbf{Deterministic detection}, ensured by fine-grained address-space asymmetry via replica-specific NOP insertion.
    \item \textbf{Probabilistic detection}, activated when execution exits the deterministic region or when perturbations are not fully correlated, and relies on structural address-space decorrelation.
\end{enumerate}

\subsection{Deterministic Detection (Fine-Grained)}

\begin{theorem}[Deterministic Detection Bound]
Let \(N\) replicas employ periodic NOP padding with stride \(l\) and replica-specific offsets
\[
\mathrm{offset}_r = \frac{l}{N} \cdot r, \quad r = 0,\dots,N-1.
\]
Assume a fully correlated perturbation:
\[
\Delta PC^{(1)} = \Delta PC^{(2)} = \cdots = \Delta PC.
\]
If
\[
|\Delta PC| \geq \frac{l}{N},
\]
then at least one replica executes a different \emph{logical} instruction, causing canonical trace divergence within one execution slice.
\end{theorem}

\begin{proof}
Replica-specific offsets induce a minimum spatial separation of \(l/N\) between corresponding logical instructions across replicas. A perturbation of magnitude \(|\Delta PC| \ge l/N\) exceeds this separation. Therefore, identical numerical perturbations cannot map all replicas to the same logical instruction. Hence, there exist replicas \(i \ne j\) executing different logical instructions, which leads to divergence of canonical traces within one execution slice.

\end{proof}

\noindent
\textit{Remark.} Experimental work on electromagnetic fault injection confirms that fully correlated
perturbations capable of skipping multiple consecutive instructions are realizable in practice~\ref{Dutertre2021}.
Thus, the adversarial model assumed in Theorem~1 is not merely theoretical.


\begin{figure}[htbp]
    \centering
    \includegraphics[width=\columnwidth]{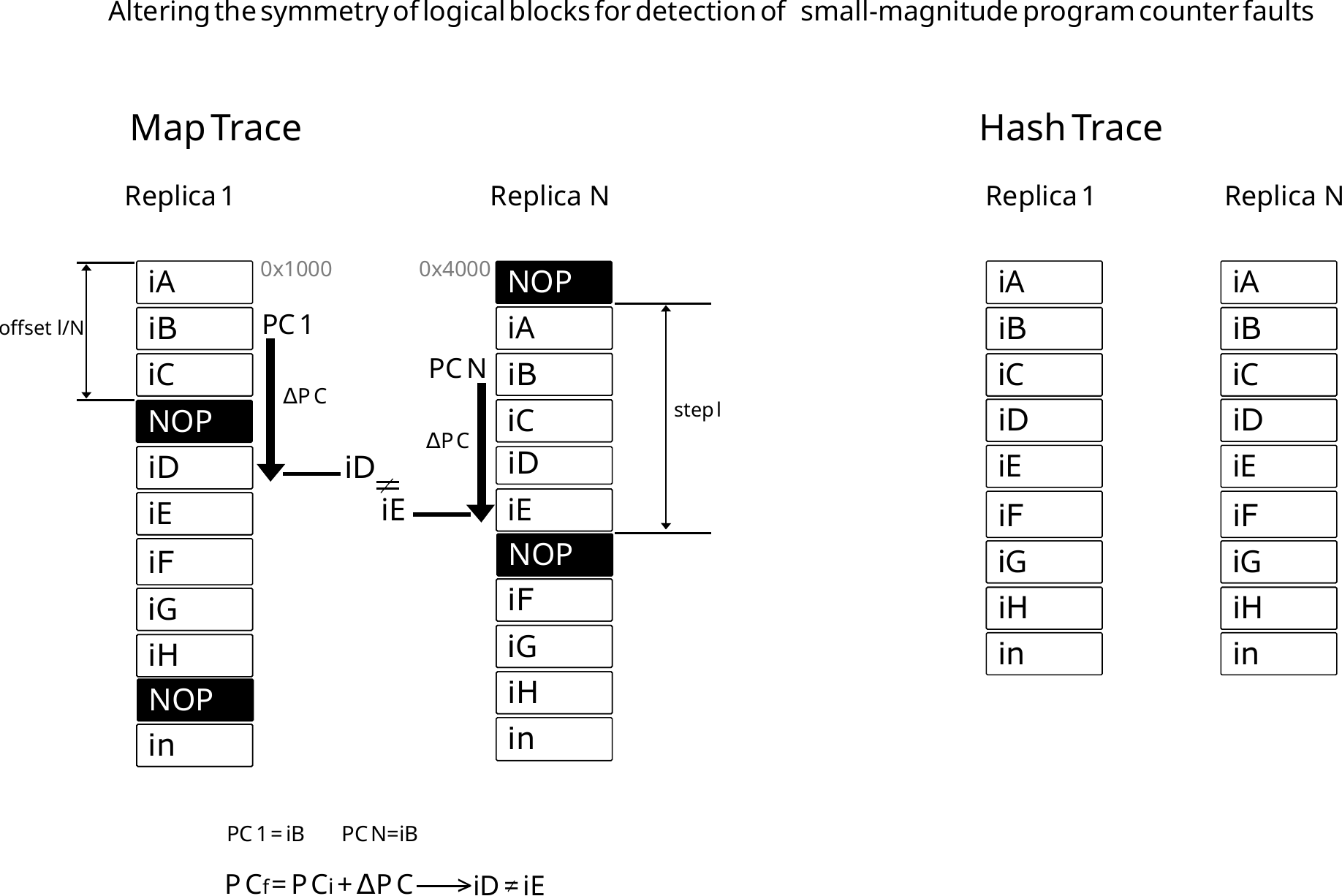}
    \caption{Deterministic NOP decorrelation: replica‑specific offsets 
         create minimum separation \(l/N\), ensuring detection of 
         correlated PC faults with \(|\Delta PC| \ge l/N\).}
    \label{fig:collapse}
\end{figure}


\paragraph{Interpretation.}
Any perturbation satisfying \(|\Delta PC| \ge l/N\) is guaranteed to be detected immediately. Perturbations with \(|\Delta PC| < l/N\) fall within an alignment window in which replicas may remain synchronized.

\paragraph{Example.}
Let \(l = 8\) bytes and \(N = 2\), so \(l/N = 4\) bytes.
\begin{itemize}
    \item \(|\Delta PC| < 4\): alignment window (no guarantee),
    \item \(|\Delta PC| \ge 4\): deterministic detection,
\end{itemize}

\subsection{Probabilistic Detection (Coarse-Grained)}

When execution exits the NOP-padded region (e.g., due to a large perturbation or corrupted control transfer), or when perturbations are not fully correlated, detection relies on divergence in decorrelated address spaces.

\paragraph{Instruction-space density.}

Let:
\begin{itemize}
    \item \( |S| \): number of valid instruction entry points,
    \item \( C \): number of instructions canonically equivalent to the expected instruction.
\end{itemize}

Define:
\[
\gamma = \frac{C}{|S|},
\]
which bounds the probability of landing on a canonically equivalent instruction.

\paragraph{Canonical equivalence.}

Each instruction \(I\) with resulting state \(s\) is mapped to:
\[
C(I, s) = (opcode, src\_regs, dst\_reg, result),
\]
excluding layout-dependent information.

\paragraph{State consistency factor.}

Let
\[
\varepsilon = \sup P(\text{equal result}),
\]
be the probability that two instructions produce identical results.

For an ISA with:
\begin{itemize}
    \item \(b_{op}\) bits for opcode,
    \item \(b_{reg}\) bits for register identifiers,
    \item \(b_{val}\) bits for result values,
\end{itemize}
we obtain:
\[
\varepsilon \le 2^{-(b_{op} + b_{reg} + b_{val})}.
\]

\paragraph{Example.}
For 8-bit opcode, three 4-bit register fields, and 32-bit result:
\[
\varepsilon \le 2^{-52}.
\]

\begin{theorem}[Probabilistic Detectability Bound]
Under structural independence, the probability that a fault remains undetected for \(k\) consecutive instructions is bounded by:
\[
P_{\mathrm{undetected}}(k) \le \left( \frac{C}{|S|} \cdot \varepsilon \right)^k.
\]
\end{theorem}

\begin{proof}
At each step, undetected execution requires:
\begin{enumerate}
    \item fetching a canonically equivalent instruction (probability \(\le C/|S|\)),
    \item producing an identical result (probability \(\le \varepsilon\)).
\end{enumerate}

Thus:
\[
P_{\text{step}} \le \frac{C}{|S|} \cdot \varepsilon.
\]

Over \(k\) steps:
\[
P_{\mathrm{undetected}}(k) \le \left( \frac{C}{|S|} \cdot \varepsilon \right)^k.
\]
\end{proof}

\paragraph{Discussion.}

The bound is conservative:
\begin{itemize}
    \item assumes uniform instruction distribution,
    \item ignores control-flow constraints,
    \item ignores state divergence amplification.
\end{itemize}

In practice, detection is significantly stronger.

\paragraph{Example.}

Let:
\begin{itemize}
    \item \( |S| = 2048 \),
    \item \( C = 12 \),
    \item \( \varepsilon \le 2^{-52} \).
\end{itemize}

For \(k = 2\):
\[
P \le \left( \frac{12}{2048} \cdot 2^{-52} \right)^2,
\]
which is negligible.

\subsection{Probabilistic Detection for Partially Correlated Faults}

The deterministic guarantee (Theorem 1) applies only to \emph{fully correlated} faults with \textbf{identical} perturbations across \emph{all} replicas, and only when \(|\Delta PC| \geq l/N\).

Two important practical cases fall outside this guarantee:

\begin{enumerate}
    \item \textbf{Differing perturbations}: A fault affects all replicas, but \(\Delta PC^{(i)} \neq \Delta PC^{(j)}\) for some \(i \neq j\).
    \item \textbf{Single-replica fault}: Only one replica experiences a perturbation; others execute correctly.
\end{enumerate}

In both cases, deterministic detection may fail due to possible \emph{re-alignment}, but detection is still possible with high probability.

\subsubsection{Detection Condition}

For a fault to remain undetected for \(k\) consecutive instructions, three conditions must hold simultaneously:

\begin{enumerate}
    \item \textbf{Equal logical instruction}: All replicas fetch instructions that are canonically equivalent:
    \[
    \Pr \leq \frac{C}{|S|}
    \]
    \item \textbf{Equal computed result}: The instructions produce identical result values:
    \[
    \Pr \leq \epsilon
    \]
    \item \textbf{No structural divergence}: Address-space collapse does not occur (i.e., PCs may differ, but the above two conditions still hold).
\end{enumerate}

\subsubsection{Probability Bound}

Under structural independence, the probability of undetected execution over \(k\) steps is bounded by the same exponential expression as in Theorem 2:

\begin{equation}
P_{\text{undetected}}(k) \leq \left( \frac{C}{|S|} \cdot \epsilon \right)^k
\end{equation}

\noindent
\textbf{Justification.}  
Even when \(\Delta PC\) values differ, the worst case for detection is when replicas accidentally align to the same logical instruction and produce identical results. The per-step probability of this coincidence is exactly \(\gamma \cdot \epsilon = \frac{C}{|S|} \cdot \epsilon\), independent of whether the perturbations were identical or not.

\subsubsection{Comparison with Lockstep}

In classical lockstep, differing perturbations or a single-replica fault cause \textbf{immediate} detection (different PCs $\rightarrow$ error flagged).

In DME, the same scenario may \emph{temporarily} escape detection if the different perturbations coincidentally map to the same logical instruction. However:

\begin{itemize}
    \item The probability of such coincidence is extremely low for realistic ISAs (\(\gamma \cdot \epsilon \approx 2^{-52}\) in the example of Section 5.2).
    \item Even if alignment occurs at step 1, the probability of maintaining it for \(k \geq 2\) becomes negligible.
\end{itemize}

Thus, DME provides \textbf{probabilistic detection} with exponential decay, whereas lockstep provides \textbf{deterministic detection} for differing perturbations.

\paragraph{Detection comparison: lockstep vs. DME}

\textbf{(A) Fully correlated, identical $\Delta PC$, $|\Delta PC| \geq l/N$:}
Lockstep: not detected (silent data corruption);
DME: deterministic detection.

\textbf{(B) Fully correlated, identical $\Delta PC$, $|\Delta PC| < l/N$:}
Lockstep: not detected;
DME: not guaranteed.

\textbf{(C) Differing $\Delta PC$ or single-replica fault:}
Lockstep: deterministic detection;
DME: probabilistic detection with bound $(\gamma \cdot \epsilon)^k$.

\subsubsection{Example}

Let \(N = 2\), \(l = 8\) bytes, \(w = 4\) bytes.  
Consider a fault affecting only Replica 0: \(\Delta PC^{(0)} = +2\), \(\Delta PC^{(1)} = 0\).

\begin{itemize}
    \item \textbf{Lockstep:} \(PC_0 \neq PC_1\) $\rightarrow$ immediate detection.
    \item \textbf{DME:} Replica 0 may land on a NOP, Replica 1 on a real instruction. Since NOP has no result, canonical equivalence is impossible — detection occurs immediately. If both land on different real instructions that accidentally produce the same result (rare), detection is delayed but bounded by \((\gamma \cdot \epsilon)^k\).
\end{itemize}

Thus, even in the probabilistic regime, practical detection is effectively immediate for most fault patterns.

\subsection{Non-Correlated Perturbations}

If perturbations differ across replicas:
\[
\Delta PC^{(i)} \ne \Delta PC^{(j)},
\]
deterministic guarantees do not apply, even within padded regions.

\paragraph{Re-alignment effect.}

Execution may temporarily re-align. For example:
\begin{itemize}
    \item Replica 1 skips a real instruction,
    \item Replica 2 skips a NOP,
\end{itemize}
resulting in both reaching the same logical instruction.

\paragraph{Example.}

\begin{itemize}
    \item Replica 1: \(\Delta PC = +3\),
    \item Replica 2: \(\Delta PC = +4\).
\end{itemize}

Both replicas may fetch the same instruction, delaying detection.

\paragraph{Detection condition.}

Undetected execution requires:
\begin{enumerate}
    \item alignment to the same logical instruction,
    \item canonical equivalence,
    \item identical computed results.
\end{enumerate}

Thus, detection follows:
\[
P_{\mathrm{undetected}}(k) \le \left( \frac{C}{|S|} \cdot \varepsilon \right)^k.
\]

\subsection{Summary of Detection Regimes}

\begin{itemize}
    \item \textbf{Deterministic regime:}  
    For fully correlated perturbations with \(|\Delta PC| \ge l/N\), detection is guaranteed within one instruction.

    \item \textbf{Alignment window:}  
    For \(|\Delta PC| < l/N\), replicas may remain aligned; no guarantee.

    \item \textbf{Non-correlated perturbations:}  
    Detection becomes probabilistic due to possible re-alignment.

    \item \textbf{Probabilistic regime:}  
    Outside padded regions or under misalignment, detection probability decays exponentially:
    \[
    \left( \frac{C}{|S|} \cdot \varepsilon \right)^k.
    \]
\end{itemize}

\paragraph{Alignment window behavior.}
If \(|\Delta PC| < l/N\) and perturbations are fully correlated, detection is not guaranteed.

However, if \(\Delta PC^{(i)} \neq \Delta PC^{(j)}\) for any \(i \neq j\), semantic divergence occurs immediately, guaranteeing detection within one slice.

Thus, the only undetectable case inside the window is identical small perturbations — a scenario equally undetectable in classical lockstep.

\section{Dual-Layer Detection}

DME employs two orthogonal detection layers:

\textbf{Semantic layer.} Compares canonical traces \(T_i\) across replicas. Detection condition: \(T_i \neq T_j\) for any \(i \neq j\) indicates a fault.

\textbf{Structural layer.} Compares physical addresses (program counter, memory addresses) across replicas. Detection condition: \(PC_1 = PC_2 = \cdots = PC_N\) indicates address-space collapse — a violation of structural independence — and triggers a system error.

Address-space collapse violates the assumption \(\phi_i \neq \phi_j\) and removes deterministic guarantees. It is detected and reported.

\subsection{Address-Space Collapse}

\textbf{Definition.} Under correct structural diversification, address mappings 
\(\phi_r\) are pairwise distinct:

\[
\forall r \neq q:\; \phi_r \neq \phi_q.
\]

\textbf{Address-space collapse} occurs if, during execution, two or more replicas 
resolve semantically corresponding operations to the same physical address:

\[
\exists i \neq j:\; PC_i = PC_j \quad \text{or} \quad \exists i \neq j:\; addr_i = addr_j.
\]

\textbf{Detection.} The structural address comparator monitors program counters 
and memory addresses across replicas. Equality of physical addresses indicates 
a violation of structural independence — a condition that cannot occur under 
correct compilation.

\textbf{Consequence.} When address-space collapse happens:
\begin{itemize}
    \item Identical numerical perturbations \(\Delta\) applied to all replicas 
          produce identical instruction fetches and memory accesses.
    \item Canonical traces remain equal: \(T_1 = T_2 = \dots = T_N\).
    \item Deterministic detectability guarantees are lost.
\end{itemize}

\textbf{System response.} Address-space collapse is treated as a system-level 
error, signalling either:
\begin{itemize}
    \item compiler-induced layout aliasing, or
    \item a fault that forced address convergence despite structural independence.
\end{itemize}

Thus, DME provides dual-layer detection: semantic divergence (canonical traces) 
and structural violation (address equality).


\begin{figure*}
    \centering
    \includegraphics[width=\textwidth]{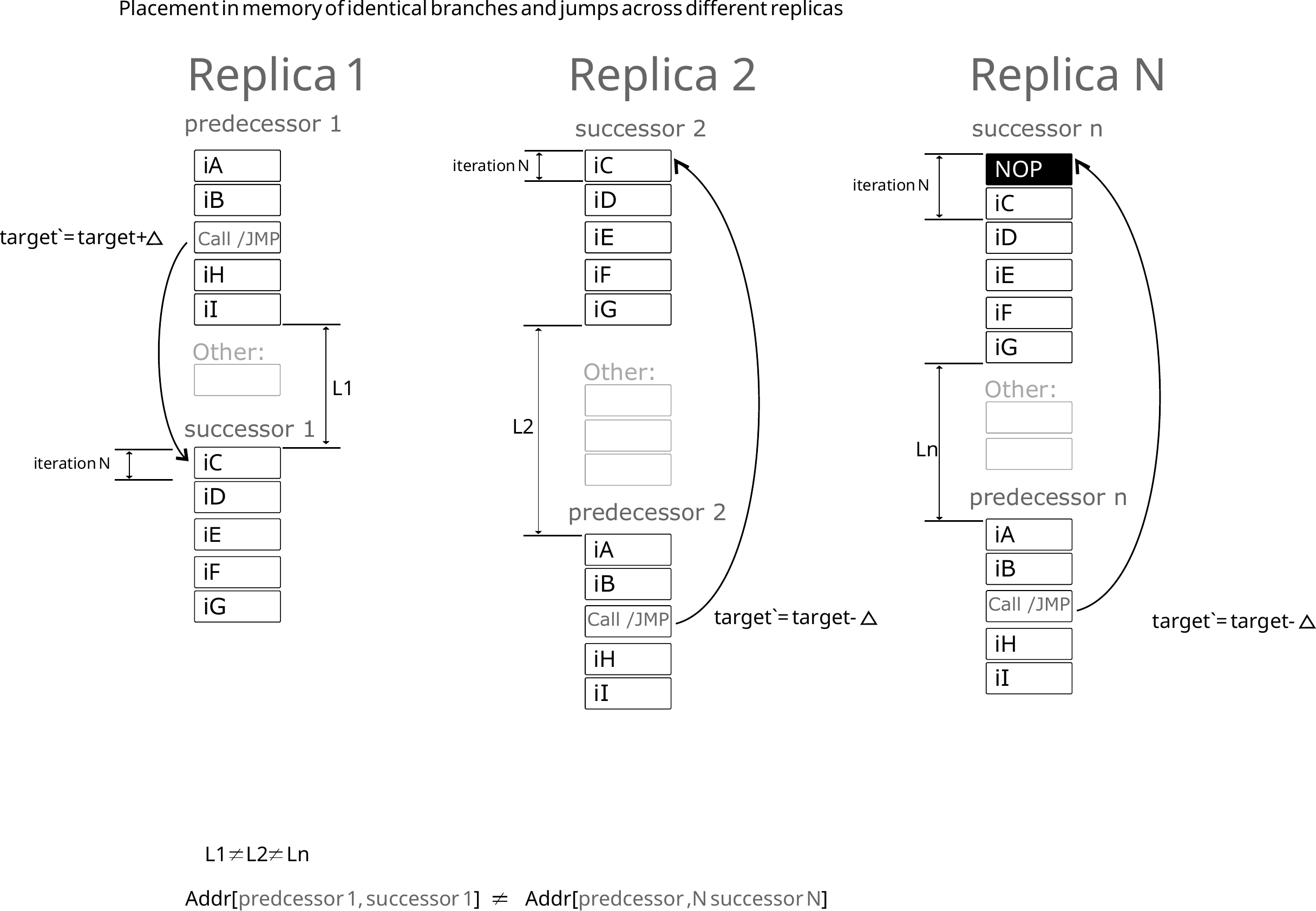}
    \caption{Example of compiler placement of three identical logical blocks and functions in memory. For instance, if a branch in the first replica has a positive offset (PC+ = 64), the same branch in a second replica may have a negative offset (PC- = 96). The delta between branch targets also differs across replicas, ensuring structural address-space diversity.}
    \label{fig:scheme}
\end{figure*}


\subsection{Address Non-Aliasing Detection (Early Memory Fault Detection)}

In addition to semantic trace comparison and structural address-space independence,
we introduce a runtime invariant that enables early detection of memory corruption
errors by comparing \emph{effective} memory access addresses across replicas. 
This includes both data pointers and instruction pointers (program counters).

\paragraph{Address Non-Aliasing Invariant.}
Under correct compilation and layout diversification, corresponding memory accesses
in different replicas must not target the same effective address. This applies to both
data accesses and instruction fetches:

\[
\forall a \in \mathcal{O},\ \forall i \neq j:\quad \phi_i(a) \neq \phi_j(a)
\]

where $\mathcal{O}$ is the set of logical memory objects (variables, stack frames,
heap allocations, \emph{and instruction addresses}), and $\phi_r$ denotes the address 
mapping function of replica $r$ that maps a logical object $a$ to its \emph{effective address}.

\textbf{Effective Address.} 
The \emph{effective address} is the address actually used for memory access or 
instruction fetch in the context of the comparison:
\begin{itemize}
    \item physical address in bare-metal systems without an MMU;
    \item virtual or replica-specific normalized address in virtualized, emulated,
          or time-multiplexed executions on a single physical core.
\end{itemize}

\paragraph{Theorem 4 (Early Detection via Address Non-Aliasing).}
Let a DME system consist of $N \ge 2$ replicas with structurally independent
address mappings $\phi_r$. Consider a memory access or instruction fetch 
(load, store, or instruction fetch) to the same logical object $a$ performed 
by all replicas within the same execution slice.

If during execution the effective addresses coincide:
\[
\phi_1(a) = \phi_2(a) = \cdots = \phi_N(a),
\]
then this constitutes a violation of the Address Non-Aliasing Invariant and indicates
a fault (or a violation of structural independence). 

In particular, if the program counters of all replicas become equal 
($\mathit{PC}_1 = \mathit{PC}_2 = \cdots = \mathit{PC}_N$), this is a severe case of 
address-space collapse and must immediately trigger an exception.

Such violations are detectable \emph{immediately} at the moment the address 
is computed or the instruction is fetched, prior to any semantic divergence 
in canonical traces.

\paragraph{Proof (informal).}
Under correct DME compilation, each replica receives a distinct address mapping
$\phi_r$. Therefore, the same logical instruction or data object resides at 
different effective addresses across replicas.

Consider a fault that corrupts a pointer or return address (for example, via 
buffer overflow) such that all replicas compute the same erroneous effective 
address:
\[
p^{(1)} = p^{(2)} = \cdots = p^{(N)} = \textit{addr}.
\]
Typical example — corruption of a return address:
\begin{verbatim}
 /* buffer overflow */
 return_address = corrupted_value;  
\end{verbatim}

Under correct structural independence such coincidence must not occur. 
Equality of effective addresses (including PCs) across replicas therefore 
violates the invariant and can be detected immediately by a dedicated 
address comparator.
\hfill $\square$

\paragraph{Corollary 4.1 (Early Detection of Return Address Corruption).}
When a return address is corrupted with the same erroneous value in all replicas 
(e.g., due to buffer overflow or systematic fault), all replicas will attempt 
to fetch the next instruction from the same effective address. This leads to 
$\mathit{PC}_1 = \mathit{PC}_2 = \cdots = \mathit{PC}_N$, which is immediately 
detected as a violation of the Address Non-Aliasing Invariant — often earlier 
than semantic trace divergence.

\paragraph{Interpretation.}
This mechanism introduces a third detection dimension in DME:

\begin{itemize}
    \item \textbf{Semantic layer}: divergence of canonical traces,
    \item \textbf{Structural layer}: enforcement of address-space independence 
          (including PC non-equality),
    \item \textbf{Address non-aliasing layer}: runtime enforcement of distinct 
          effective access targets for both data and instructions.
\end{itemize}

The address non-aliasing layer is particularly effective against correlated 
pointer and control-flow faults (such as identical return address corruption), 
providing zero-latency detection based purely on structural violations.

\paragraph{Limitations.}

\begin{itemize}
    \item \textbf{Dependence on layout guarantees.}
    The invariant requires strict enforcement of non-overlapping address mappings.
    Any compiler or linker behavior that produces identical mappings for some
    addresses may lead to false positives.

    \item \textbf{Legitimate shared regions.}
    Memory regions intentionally shared across replicas (e.g., memory-mapped I/O,
    communication buffers) violate the invariant and must be excluded.

    \item \textbf{Partial coverage.}
    The mechanism detects only faults that result in identical logical addresses.
    If corrupted pointers differ numerically, detection falls back to semantic
    divergence.

    \item \textbf{No semantic validation.}
    Address equality or inequality does not directly encode program correctness.
    The mechanism detects violations of structural assumptions rather than logical errors.

    \item \textbf{Hardware overhead.}
    Continuous comparison of memory access addresses introduces additional hardware
    cost and may impact timing in constrained systems.
\end{itemize}

\paragraph{Summary.}
Address Non-Aliasing Detection strengthens DME by transforming a subset of latent
memory corruption faults into immediately observable structural violations,
reducing detection latency and complementing canonical trace comparison.


\subsection{Detection of Pointer-Semantics Violations}

Structural address-space decorrelation in DME requires
that semantically corresponding objects reside at
different effective addresses in different replicas.

Consequently, pointers referring to corresponding logical
objects must also differ numerically across replicas:

\[
\forall i \neq j : p_i(a) \neq p_j(a)
\]

where \(p_r(a)\) denotes the pointer value referencing
logical object \(a\) in replica \(r\).

This property enables detection of a broad class of
software errors involving invalid pointer semantics.

\paragraph{Null-pointer collapse.}

Consider a fault or software defect that replaces a valid
pointer with zero:

\[
p_1 = p_2 = \cdots = p_N = 0
\]

When replicas attempt to dereference the pointer,
all replicas access the identical effective address:

\[
addr_1 = addr_2 = \cdots = addr_N = 0
\]

This violates the Address Non-Aliasing Invariant and
immediately triggers detection as an address-space
collapse event.

Importantly, detection occurs even if the fault originates
purely from software logic rather than from hardware
corruption.

\paragraph{Value-as-pointer corruption.}

Consider a programming error in which a computed value
is mistakenly stored into a pointer variable:

\[
p := A + B
\]

Since arithmetic results are semantically equivalent
across replicas during correct execution, all replicas
compute the same numerical value:

\[
(A+B)_1 = (A+B)_2 = \cdots = (A+B)_N
\]

Consequently, the corrupted pointer values also become
identical:

\[
p_1 = p_2 = \cdots = p_N
\]

A subsequent load/store operation through the corrupted
pointer causes all replicas to access the same effective
address, violating structural independence and producing
an address-space collapse.

Thus, DME transforms certain classes of latent software
bugs into immediately observable structural violations.

Examples include:

\begin{itemize}
    \item null-pointer dereference,
    \item use of uninitialized pointers,
    \item accidental value-to-pointer assignment,
    \item corrupted function pointers,
    \item invalid return addresses,
    \item pointer truncation or overflow,
    \item pointer arithmetic errors.
\end{itemize}

Interpretation. Address-space decorrelation is therefore
not only a mechanism for correlated fault detection, but
also a runtime validator of pointer semantic correctness.


\section{Formal Correctness Criteria in DME}

Conventional redundancy mechanisms such as lockstep execution and Triple Modular Redundancy (TMR) define correctness solely in terms of equality of architectural states across replicas. While effective against independent faults, this definition admits a class of correlated faults that preserve equality while violating the intended program semantics.

Divergent Multi-Version Execution (DME) extends this notion by explicitly separating two orthogonal aspects: (i) correctness of execution and (ii) conditions required for reliable fault observability. Rather than treating all deviations uniformly, DME introduces a two-level model consisting of a runtime semantic correctness condition and a design-time structural constraint.

\subsection{Semantic Equivalence}

Let \(T_r\) denote the canonical execution trace of replica \(r\), defined as a sequence of layout-independent representations of executed instructions. We define semantic equivalence as:
\[
P_{\text{sem}} = (T_1 = T_2 = \cdots = T_N) \tag{1}
\]
This condition ensures that all replicas execute logically identical instruction sequences and produce identical semantic effects.

Importantly, \(P_{\text{sem}}\) is agnostic to the origin of deviations. Any violation of semantic equivalence captures incorrect execution regardless of whether it arises from transient hardware faults, memory corruption, or compiler-induced errors.

\subsection{Structural Independence}

Let \(\phi_r\) denote the address mapping function for replica \(r\), assigning logical program elements to physical memory addresses. We define structural independence as:
\[
P_{\text{str}} = (\forall r \neq q : \phi_r \neq \phi_q) \tag{2}
\]

This condition ensures that replicas do not share identical memory layouts, eliminating structural correlation in address space.

Structural independence is established at compile time through replica-specific transformations and is invariant during execution. Its purpose is not to define correctness directly, but to ensure that identical numerical perturbations across replicas do not result in identical logical behavior.

\subsection{Correctness Model}

We distinguish between runtime correctness and system-level correctness.

Runtime correctness is defined solely in terms of semantic equivalence:
\[
\text{Correct}_{\text{runtime}} = P_{\text{sem}} \tag{3}
\]

System correctness additionally requires structural independence:
\[
\text{Correct}_{\text{system}} = P_{\text{sem}} \land P_{\text{str}} \tag{4}
\]

This distinction reflects the separation between execution behavior and the conditions that guarantee its reliable verification.

\subsection{Error Condition}

A runtime error is detected when semantic equivalence is violated:
\[
\text{Error} = \neg P_{\text{sem}} \tag{5}
\]

Violation of structural independence does not correspond to a runtime fault, but rather to a configuration or design-time error that may reduce the detectability of faults.

\subsection{Interpretation}

This formulation separates two complementary dimensions:
\begin{itemize}
\item Semantic correctness (\(P_{\text{sem}}\)), capturing all deviations in execution behavior, including control-flow errors, data corruption, and miscompilation.
\item Structural independence (\(P_{\text{str}}\)), ensuring that such deviations manifest as observable divergence between replicas.
\end{itemize}

In contrast to conventional redundancy schemes, which rely solely on state equivalence, DME explicitly models the conditions under which faults become observable. Structural independence does not define correctness by itself, but guarantees that violations of semantic correctness are detectable.

\subsection{Discussion}

Structural independence is enforced at compile time through replica-specific layout transformations and remains invariant during execution. It serves as a prerequisite for fault detectability rather than a runtime correctness condition.

Intuitively, DME narrows the correctness criterion to semantic equivalence while simultaneously preventing fault masking due to structural correlation. As a result, faults originating from diverse sources — including hardware faults, memory corruption, and compiler errors — are uniformly exposed as semantic divergence in canonical execution traces.

This dual perspective allows DME to be interpreted not only as a fault detection mechanism, but also as a general runtime validator of semantic equivalence across diversified execution pipelines.

\subsection{Detectability Guarantee under Structural Independence}

The role of structural independence can be formalized as a condition that guarantees observability of semantic deviations under correlated perturbations.

\textbf{Fault model.} Assume a fault that induces identical numerical perturbations \(\Delta\) across all replicas, affecting either:
\begin{itemize}
\item the program counter (control-flow faults), or
\item data values or memory addresses (data-path faults).
\end{itemize}
Such faults may arise from global disturbances (e.g., voltage fluctuations, clock glitches) or systematic effects.

\begin{theorem}[Detectability under structural independence]
Let a DME system satisfy structural independence \(P_{\text{str}}\). Assume that the address mapping functions \(\phi_r\) are pairwise distinct and that canonicalization removes all layout-dependent information.

If a fault induces identical perturbations \(\Delta\) in all replicas and the perturbation leads to a semantic deviation in at least one replica, then there exist replicas \(i \neq j\) such that:
\[
T_i \neq T_j \tag{6}
\]
and therefore the fault is detectable via violation of \(P_{\text{sem}}\) within a finite number of execution steps.
\end{theorem}

\begin{proof}[Proof (informal)]
Because structural independence ensures \(\phi_i \neq \phi_j\) for all \(i \neq j\), identical numerical perturbations \(\Delta\) applied to different replicas are mapped through distinct address spaces. Consequently, the same perturbed value (e.g., program counter or pointer) resolves to different logical instructions or data locations in different replicas.

This divergence in logical behavior produces different canonical instruction representations \(C(I_t, s_t)\) across replicas. Since canonicalization removes layout-dependent artifacts while preserving semantic effects, any difference in executed instructions, loaded values, or computed results results in different canonical traces.

Therefore, for at least two replicas \(i\) and \(j\), the traces diverge (\(T_i \neq T_j\)), which violates \(P_{\text{sem}}\) and is detected by the comparison mechanism. Detection occurs within a bounded number of execution steps determined by the trace comparison granularity. $\Box$
\end{proof}

\paragraph{Stronger formulation.}
The same argument extends to data-path faults. Suppose a fault induces
identical perturbations to register values or memory operands across
all replicas. Because structural independence ensures $\phi_i \neq
\phi_j$, the \emph{same} perturbed numerical value maps to different
logical data locations. Consequently, loaded values differ across
replicas, producing different computed results in ALU operations.
These differences appear in the canonical traces as divergent
$computed\_result$ fields, triggering detection.

Thus, structural independence guarantees not only control-flow
divergence but also \emph{register state divergence} within a bounded
number of steps. The only way for a fault to remain undetected is to
preserve \emph{identical register states and identical instruction
sequences} indefinitely — a scenario whose probability is bounded by
$\rho^{N-1} \cdot p^k$ with $p$ extremely small for $k \ge 2$.


\begin{figure*}
    \centering
    \includegraphics[width=\textwidth]{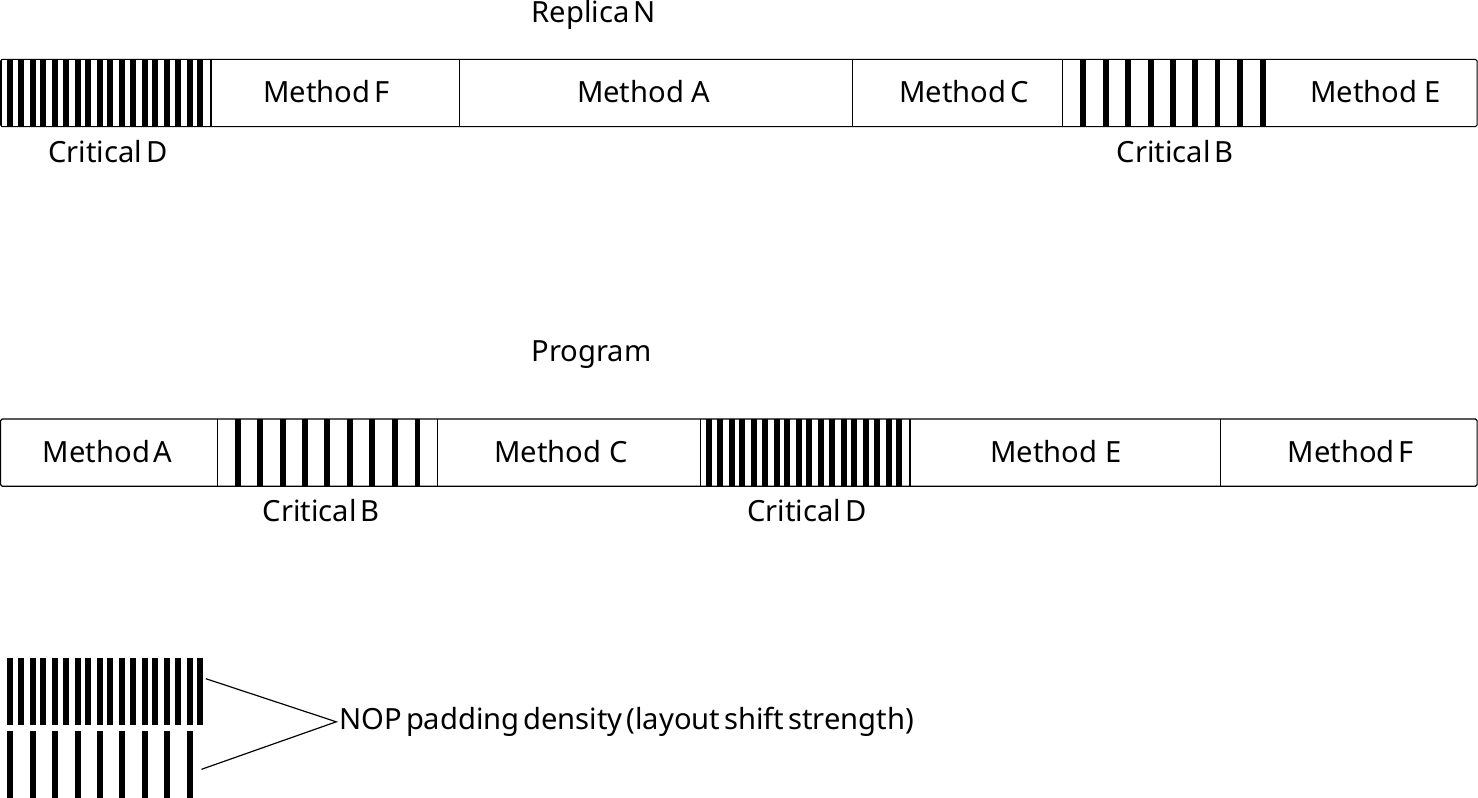}
    \caption{Adaptive DME layout diversification.
Fine-grained instruction padding is applied only
to critical program regions, while the rest of the
code relies on coarse-grained structural addressspace decorrelation. The density of vertical
stripes schematically represents the density of inserted NOP instructions, illustrating increased finegrained layout diversification in critical functions.}
    \label{fig:scheme}
\end{figure*}


\begin{corollary}
In the absence of structural independence (i.e., if \(\exists i \neq j : \phi_i = \phi_j\)), there exists a class of correlated faults for which:
\[
T_1 = T_2 = \cdots = T_N \tag{7}
\]
despite incorrect execution, leading to potential silent data corruption.
\end{corollary}

\textbf{Interpretation.} The theorem establishes that structural independence is a sufficient condition for transforming correlated low-level perturbations into observable semantic divergence. In this sense, structural diversity acts as a fault amplification mechanism: identical numerical faults are mapped to different logical effects, ensuring their detectability.

Thus, while semantic equivalence defines correctness, structural independence guarantees that violations of correctness cannot remain hidden due to correlated execution.

\textbf{Remark (Limit of detectability).} The detectability guarantee established above assumes that identical perturbations do not preserve semantic equivalence across all replicas. If a fault produces identical canonical instruction effects in every replica, then:
\[
T_1 = T_2 = \cdots = T_N
\]
and no divergence is observable. This limitation is inherent to all redundancy-based detection mechanisms, including lockstep execution and TMR. DME mitigates this limitation by introducing structural diversity, which reduces the class of faults that preserve semantic equivalence, but does not eliminate it entirely.

\subsection{Runtime Enforcement of Structural Independence}

While structural independence \(P_{\text{str}}\) is established at compile time, runtime address-level consistency checks (e.g., verifying \(PC_i \neq PC_j\) for semantically equivalent instructions or comparing section base addresses) provide partial enforcement, detecting address-space collapse as defined in Section~4.

\subsection{Execution Slice}

\textbf{Definition.} An \textit{execution slice} \(\mathcal{S}_t\) is the interval 
during which each of the \(N\) replicas executes exactly one logical instruction 
before the canonical hashes are compared. The physical duration of a slice is:

\[
T_{\mathrm{slice}} = N \cdot T_{\mathrm{instr}} + T_{\mathrm{cmp}}
\]

in a time-multiplexed implementation, or simply \(T_{\mathrm{slice}} = T_{\mathrm{instr}}\) 
in a spatially parallel (e.g., FPGA) implementation.

\textbf{Property.} Any fault that causes semantic divergence is detected within 
at most one execution slice. Interrupts and exceptions are handled identically 
across replicas before slice boundaries.

\section{Limitations}

\DME inherits the fundamental limitation of all redundancy-based schemes:

\emph{Faults that preserve identical canonical traces across all replicas remain undetectable.}

Examples:
\begin{itemize}
\item Perturbations that map to equivalent logical instructions
\item Faults affecting only shared microarchitectural components identically
\item Compiler-induced correlated transformations that preserve semantics
\end{itemize}
\DME reduces the probability of such masking but does not eliminate it.

\paragraph{Remark on the undetectable class.}
While undetectable faults theoretically exist, our analysis shows that
they require the preservation of \emph{identical register states}
across all replicas throughout the faulty execution. For a program
with $n$ instructions, the number of distinct register states is
exponential in the number of registers and bit width. The probability
that a random PC perturbation lands on an instruction sequence that
preserves register equality for $k$ steps is bounded by
$\rho^{N-1} \cdot \epsilon^k$, where $\epsilon$ is the probability
that two randomly chosen instructions produce the same result given
identical inputs. For typical ALU operations on 32-bit data,
$\epsilon \approx 2^{-32}$ for addition/subtraction, and even smaller
for operations with carry or flags.

Hence, while not strictly impossible, the undetectable fault class is
\emph{practically negligible} for any $k \ge 2$. This distinguishes
DME from lockstep/TMR, where undetectable correlated faults are not
exponentially suppressed by state width.

\section{Conclusion}

Divergent Multi-Version Execution redefines redundancy: from comparing physical state to comparing logical execution trajectories.

\textbf{Key results:}
\begin{itemize}
\item Deterministic detection for \(|\Delta PC| \geq l/N\) — provable, not probabilistic
\item Dual-layer detection (semantic + structural) catches both divergence and address-space collapse
\item Practical implementations: 
    \begin{itemize}
        \item VM on Cortex-M: 8 KB RAM total, VM occupies 4 KB, remaining 4 KB shared among replicas (2 KB per replica for \(N=2\))
        \item FPGA implementation on Altera Cyclone IV EP4CE6E22C8 operating in dual-core mode. Each core consumes 3000 LUTs and 16 KB of BRAM per replica. Soft processor, designated \textit{3o|||sheet}, features a custom 32-bit instruction set designed for industrial controller workloads
    \end{itemize}
\end{itemize}
DME provides a formally analyzable, configurable, and hardware-independent fault detection mechanism suitable for safety-critical embedded systems.

\subsection{Relation to Multi-Variant Execution Environments (MVEE)}

Unlike Multi-Variant Execution Environments (MVEE), which primarily employ diversification to reduce exploit portability while validating coarse-grained external behavior (e.g., system calls), DME uses structural address-space decorrelation as a fault-observability mechanism.

In MVEE systems, diversified layouts are mainly a security-oriented property and do not generally constitute a formal runtime correctness invariant. Correlated instruction-level perturbations may therefore remain undetected as long as externally observable behavior remains sufficiently consistent across variants.

In contrast, DME is explicitly designed to detect bit-level corruptions of the program counter caused by physical disturbances (e.g., electromagnetic or clock glitches). Such faults typically shift \(PC\) by a small offset rather than overwriting it with an arbitrary value. To make these small, correlated shifts observable, DME does not rely on passive layout randomness. Instead, it actively modifies branch displacements during compilation: the same logical branch is compiled with opposite signs across replicas (e.g., forward in replica 0, backward in replica 1). Consequently, an identical \(\Delta PC\) perturbation pushes replicas onto semantically different execution paths, guaranteeing detection.

More generally, DME continuously validates semantic equivalence of canonical instruction traces and treats violations of structural independence as explicit runtime fault conditions. These violations include not only control-flow collapse:
\[
PC_1 = PC_2 = \dots = PC_N,
\]
but also data address collapse:
\[
p_1 = p_2 = \dots = p_N,
\]
where \(p_r\) denotes a data pointer in replica \(r\). A null pointer dereference, a corrupted return address, or an arithmetic result mistakenly assigned to a pointer — if identical across all replicas — is detected immediately via the Address Non-Aliasing Invariant, often before any semantic divergence occurs.

Thus, structural diversity in DME is not merely intended to increase exploit diversity, but to guarantee that identical numerical perturbations cannot preserve identical logical execution across replicas — for both code and data addresses.

\subsection{DME vs. CFI}

\begin{itemize}
    \item \textbf{CFI:} single execution, enforces legal control-flow transitions (CFG), cannot detect correlated faults that stay within the CFG or data-pointer errors without illegal jumps.
    \item \textbf{DME:} multiple diversified replicas, compares canonical instruction traces (opcodes, registers, results, load/store values), detects correlated PC faults ($|\Delta PC|\ge l/N$) and identical pointer corruptions via address non-aliasing invariant.
    \item \textbf{Relationship:} CFI $\subseteq$ DME per replica; DME adds cross-replica semantic equivalence and structural address invariants.
\end{itemize}

\subsection{Distinction from EDDI and SWIFT}

EDDI execute identical binaries with identical memory layouts, detecting only \emph{uncorrelated} faults. DME executes independently compiled replicas with decorrelated address spaces. Consequently, DME provides:
\begin{itemize}
    \item Deterministic detection of fully correlated PC faults when \(|\Delta PC| \ge l/N\) (Theorem~1);
    \item Immediate structural violation detection for identical pointer corruptions via the Address Non-Aliasing Invariant.
\end{itemize}
Neither property holds in EDDI or SWIFT.

\section*{Selected References}
\begin{enumerate}
\item S. Manoni et al., ``CVA6-CFI: A First Glance at RISC-V Control-Flow Integrity Extensions,'' \textit{arXiv preprint arXiv:2602.04991}, Feb. 2026. \label{Manoni2026}
\item A. Avizienis, ``The N-Version Approach to Fault-Tolerant Software,'' \textit{IEEE Trans. Software Eng.}, 1985. \label{Avizienis1985}
\item N. Oh et al., ``Error Detection by Duplicated Instructions (EDDI),'' \textit{MICRO-35}, 2002. \label{Oh2002}
\item G. A. Reis et al., ``SWIFT: Software Implemented Fault Tolerance,'' CGO, 2005. \label{Reis2005}
\item M. Abadi et al., ``Control-Flow Integrity,'' CCS, 2005. \label{Abadi2005}
\item J.-M. Dutertre, A. Menu, O. Potin, J.-B. Rigaud, J.-L. Danger, ``Experimental Analysis of the Electromagnetic Instruction Skip Fault Model and Consequences for Software Countermeasures,'' \textit{Microelectronics Reliability}, 2021. \label{Dutertre2021}
\item B. Randell, ``System structure for software fault tolerance,'' \textit{IEEE Transactions on Software Engineering}, vol. SE-1, no. 2, pp. 220-232, June 1975. \label{Randell1975}
\item J. Just and M. Cornwell, ``Review and analysis of synthetic diversity for breaking monocultures,'' in \textit{Proc. 2004 ACM Workshop on Rapid Malcode (WORM '04)}, Washington DC, USA, 2004, pp. 23-32. \label{Just2004}
\item H.-M. Pham, S. Pillement, and S. J. Piestrak, ``Low-overhead fault-tolerance technique for a dynamically reconfigurable softcore processor,'' \textit{IEEE Transactions on Computers}, vol. 62, no. 6, pp. 1202-1215, June 2013. \label{Pham2013}
\item M. Barbiottia et al., ``Dynamic Triple Modular Redundancy in Interleaved Hardware Threads,'' \textit{IEEE Access}, vol. 12, pp. 32456-32470, 2024. \label{Barbiottia2024}
\item S. Binosi et al., ``The Illusion of Randomness: An Empirical Analysis of Address Space Layout Randomization Implementations,'' in \textit{Proc. ACM SIGSAC Conf. Computer and Communications Security (CCS '24)}, 2024. \label{Binosi2024}
\item A. Avizienis, J. C. Laprie, B. Randell, and C. Landwehr, ``Basic concepts and taxonomy of dependable and secure computing,'' \textit{IEEE Transactions on Dependable and Secure Computing}, vol. 1, no. 1, pp. 11-33, Jan.-March 2004. \label{Avizienis2004}
\item M. Amel Solouki, S. Angizi, and M. Violante, ``Dependability in Embedded Systems: A Survey of Fault Tolerance Methods and Software-Based Mitigation Techniques,'' \textit{arXiv preprint arXiv:2404.10509}, April 2024. \label{Solouki2024}
\item A. A. Malik, H. Mihir, and A. Aysu, ``Honest to a Fault: Root-Causing Fault Attacks with Pre-Silicon RISC Pipeline Characterization,'' \textit{arXiv preprint arXiv:2503.04846}, March 2025. \label{Malik2025}
\item ``Improved address space layout randomization,'' Google Patents DE112017002277T5, 2017. \label{GooglePatent2017}
\end{enumerate}

\end{document}